\documentclass{amsart}
 \usepackage[foot]{amsaddr}
\usepackage[top=4.4cm,bottom=4.4cm,left=3.4cm,right=3.4cm]{geometry}               
\usepackage{amsmath,amscd,amssymb,amsthm}
\usepackage[english]{babel}
\usepackage{mathtools}
\usepackage{enumerate}
\usepackage{setspace}
\usepackage{mathrsfs,makecell}
\usepackage{color,array}
\usepackage[numbers]{natbib}
\usepackage[hidelinks]{hyperref}
\usepackage{subcaption}
\usepackage{caption}
\usepackage{tabularx}

\usepackage{subcaption}
\usepackage{tabularx}
\usepackage{amsmath,amssymb,amsthm}
\usepackage[english]{babel}
\usepackage{mathtools}
\usepackage{enumerate}
\usepackage[numbers]{natbib}
\usepackage[hidelinks]{hyperref}
\usepackage{tikz}
\usepackage{tkz-graph}
\usepackage{tikz-cd}
\tikzset{%
    symbol/.style={%
        ,draw=none
        ,every to/.append style={%
            edge node={node [sloped, allow upside down, auto=false]{$#1$}}}
    }
}
\usepackage{xcolor}

\usepackage{algorithm}
\usepackage[noend]{algpseudocode}

\DeclareMathOperator{\Char}{char}

\DeclareMathOperator{\Gal}{Gal}

\DeclareMathOperator{\Ram}{Ram}
\DeclareMathOperator{\Id}{Id}

\def\vF{\mathbb{F}}
\def\vN{\mathbb{N}}
\def\vZ{\mathbb{Z}}

\def\cO{\mathcal{O}}
\def\cS{\mathcal{S}}

\theoremstyle{definition}
\newtheorem{definition}{Definition}[section]
\newtheorem{remark}[definition]{Remark}

\theoremstyle{plain}
\newtheorem{theorem}[definition]{Theorem}

\newtheorem{proposition}[definition]{Proposition}

\newtheorem*{theorem*}{Theorem}

\newtheorem{problem}[definition]{Problem}

\makeatletter
\renewcommand{\email}[2][]{%
  \ifx\emails\@empty\relax\else{\g@addto@macro\emails{,\space}}\fi%
  \@ifnotempty{#1}{\g@addto@macro\emails{\textrm{(#1)}\space}}%
  \g@addto@macro\emails{#2}%
}
\makeatother

\makeindex

\author[Giacomo Micheli]{Giacomo Micheli}
\address{Mathematical Institute\\
University of Oxford\\
Woodstock Rd\\ 
Oxford OX2 6GG, United Kingdom
}
\title[Constructions of Locally Recoverable Codes which are Optimal]{Constructions of Locally Recoverable Codes which are Optimal}

\thanks{The author was supported by the Swiss National Science Foundation grant number 171248 and 171249.}


\begin{document}

\begin{abstract}
Let $q$ be a prime power and $\vF_q$ be the finite field of size $q$. In this paper we provide a Galois theoretical framework that allows to produce good polynomials for the Tamo and Barg construction of optimal locally recoverable codes (LRC). Using our approach we construct new good polynomials and then optimal LRCs with new parameters. The existing theory of good polynomials fits entirely in our new framework. The key advantage of our method is that we do not need to rely on arithmetic properties of the pair $(q,r)$, where $r$ is the locality of the code.
\end{abstract}

\maketitle

\section{Introduction}
Let $n,k,r$ be positive integers with $k\leq n$.
A locally recoverable code (LRC) $\mathcal C$ having parameters $(n,k,r)$ is an $\vF_q$-subspace of $\vF_q^n$ of dimension $k$ such that, if one deletes one component of any $v\in \mathcal C$, this can be recovered by accessing at most $r$ other components of $v$. LRCs are of great interest in the context of distributed and cloud storage systems \cite{ballico2016higher,barg2015locally,
cadambe2013upper,gopalan2014explicit,papailiopoulos2013repair,
prakash2012optimal,silberstein2013optimal,tamo2014bounds,tamo2016optimal}. One of the most interesting constructions of LRC is due to Tamo and Barg \cite{tamo2014family} and is realised via constructing polynomials of degree $r+1$ which are constants on subsets of $\vF_q$ of cardinality $r+1$. These polynomials are called \emph{good} polynomials.
Constructions of good polynomials are also provided in \cite{liu2018new}. All these constructions are mostly based on arithmetic properties of the pair $(q,r)$ (for a survey of previous constructions and some new ones, all with such restrictions, see \cite[Section III]{liu2018new}). One of the purposes of this paper is to drop these constraints using a general density approach based on algebraic number theory.

As often happens in coding theory, when one is allowed to go over extension fields problems become easier: for example, in the case in which $r+1$ is coprime with $q$, one is always able to fix a good monomial of degree $r+1$ as long as one fixes an extension $\vF_{q^d}$ such that $r+1\mid q^d-1$ (but in this case $d$ might even have size similar to the one of $r$ when $q$ has large multiplicative order in $(\vZ/(r+1)\vZ)^*$).
Also, it should be noticed that, whenever available (i.e. when $r+1\mid q-1$) the monomial $x^{r+1}$ is one of the optimal choices, as the sets of size $r+1$ where $x^{r+1}$ is constant cover entirely $\vF_q^*$ (see also Remark \ref{rem:monomial}).  One the challenges here is to fix a certain base field $\vF_q$ and construct an optimal code with locality $r$, and large length and dimension. To achieve this using the construction of Tamo and Barg, we need a good polynomial of degree $r+1$ that is constant on as many subsets of $\vF_q$ as possible (this will determine the range of length and dimensions that are available for a fixed base field $\vF_q$ and a given locality $r$).

In this paper we fit the theory of good polynomials in a Galois theoretical context, showing that finding polynomials that are good can be reduced to solve a Galois theory problem. Moreover, existing constructions of good polynomials fit completely in our context and can be derived by our main theorems in Section \ref{sec:GalThFF}.

In addition, using the same method we provide existential results in Subsection \ref{sec:existencegood} and explicit new constructions with new parameters in Subsection \ref{sec:selection}.
In Section \ref{sec:examples} we put the method in practice constructing good polynomials for many different base fields. All the results are based on estimates provided by the Chebotarev Density Theorem for global function fields, which are also supported by the tables in Section \ref{sec:examples}.

We explain now in a nutshell the key idea of our constructions.
We are interested in polynomials $f$ of degree $r+1$ which are constants on disjoint sets $A_i\subset \vF_q$ (i.e. $f(A_i)=\{t_i\}$, for some $t_i\in \vF_q$) each of cardinality $r+1$, for $i\in\{1,\dots \ell\}$, for some $\ell\in \vN$.  One of the tasks is to maximise $\ell$, as this will allow to build LRCs with large length and dimension, as it is explained in Theorem \ref{thm:LRCconstruction} (of course, it is trivial to get $\ell=1$, as it is enough to take a totally split polynomial).
The fundamental observation is that $\ell$ is exactly the number of totally split places of degree $1$ of the global function field extension $\vF_q(x):\vF_q(t)$, where $t$ is trascendental over $\vF_q$ and $x$ is a zero of $f-t$ over the algebraic closure of $\vF_q(t)$.
Now, a place is totally split in $\vF_q(x):\vF_q(t)$ if and only if it is totally split in $M_f:\vF_q(t)$, where $M_f$ is the Galois closure of the extension $\vF_q(x):\vF_q(t)$. Let $G_f$ be the Galois group of $M_f:\vF_q(t)$.
By the Chebotarev Density Theorem (and under some additional technical requirements), the number of totally split places is roughly of the size of $q/\# G_f$. Therefore, the entire task of finding good polynomials relies on finding polynomials $f\in \vF_q[X]$ with minimal Galois group $G_f=\Gal(M_f:\vF_q(t))$. 

The strenght of the method we propose relies on the fact that we do not need to base our constructions on algebraic properties of the base field as in \cite{liu2018new,tamo2014family} but we extract good polynomials with a density argument. Let us now show with a toy example a class of good polynomials that deceives the constructions in \cite{liu2018new,tamo2014family}.
Suppose that we want to construct via good polynomials a $(n,k,2)$-LRC over an alphabet of size $q$, with $q\equiv 2 \mod 3$ using one of the constructions in \cite{liu2018new}. One would need a degree $3$ polynomial that is totally split at least at some place $t_0\in \vF_q$. But then this cannot be a composition, as its degree is prime and cannot be a linearised polynomial, as its degree is incompatible with the characteristic. Also, it cannot be a power because $x\mapsto x^3$ is a bijection by the congruence class of $q$ modulo $3$, so none of the constructions in \cite[Section III]{liu2018new}  apply. Of course, one could get an $(n,k,2)$-LRC by extending the base field to $\vF_{q^2}$ (instead of $\vF_q$) and using the good polynomial $x^3$. The approach we propose in this paper does not need field extensions and allows to construct good polynomials directly over the field where we want the code: for this specific case $r=2$ see for example Theorem \ref{thm:degree_3_pred}, which always ensures the existence (and constructibility) of a good polynomial of degree $3$ with predictable large $\ell$ without increasing the field size. To show how effective the method is, in the rest of the paper we construct explicitly many good polynomials over many different fields where previous constructions are not available (see for example Proposition \ref{thm:pessimistic} and Theorems  \ref{thm:degree_3_pred}, \ref{thm:x4x2prediction}, \ref{thm:degree_6_pred}).
%

\subsection*{Notation}
For us, a  $(n,k,r)$-LRC is a subspace of $\vF_q^n$ of dimension $k$ with locality $r$, i.e. if a component of a codeword is lost, then it can be recovered by looking at most at $r$ other components.
Let $F,K$ be fields. We denote by $F[X]$ the polynomial ring in the variable $X$ over the field $F$. A field extension $K\subseteq F$ will be denoted by $F:K$ (and not $F/K$, in order not to overlap with other notation) and its degree by $[F:K]$, i.e. the dimension of $F$ as a $K$-vector space. For any $\alpha\in K$ that is not a square, we denote by $\sqrt{\alpha}$ any zero of $X^2-\alpha$ over the algebraic closure of $F$.
Let $p$ be a prime number, $m$ a positive integer, $q=p^m$  and  $\vF_q$ be the finite field of order $q$.
Let $t$ be trascendental over $\vF_q$ and $\vF_q(t)$ be the rational function field in the variable $t$ over the base field $\vF_q$. 
In this paper we use the notation and terminology of \cite{stichtenoth2009algebraic}, which we briefly recall here.
A global function field $F$ over $\vF_q$ is a finite dimensional extension of $\vF_q(t)$. For a place $P$ of $F$, we denote by $\cO_P$ its valuation ring and say that $P$ has degree $1$ if $[\cO_P/P:\vF_q]=1$. We will only deal with global function fields, so the term global will mostly be understood.

For a function field $F$ over $\vF_q$, we denote by $\mathcal P(F)$ (resp. $\mathcal P^1(F)$) the set of places (resp. places of degree $1$) of $F$.
Let $F:K$ be an extension of function fields. Let $P\subseteq K$ be a place of $K$ and $Q\subseteq F$ be a  place of $F$. We say that $P$ lies above $Q$ if $P\subseteq Q$. Moreover we denote by $e(Q\mid P)$ (resp. $f(Q\mid P)$) the ramification index (resp. the relative degree) of the extension of places $Q\mid P$ . 
We say that $f\in \vF_q[X]$ is a separable polynomial if $f\notin \vF_q[X^p]$, in such a way that $f-t$ is a separable irreducible polynomial over $\vF_q(t)$.
We denote by $M_f$ the splitting field of $f-t$ over $\vF_q(t)$, i.e. the Galois closure of the extension $\vF_q(x):\vF_q(t)$, where $x$ is any of the roots of $f-t$ over the algebraic closure of $\vF_q(t)$. We denote by $k_f$ the field of constants of $M_f$, i.e. 
\[k_f=\{u\in M_f: \, \text{$u$ is algebraic over $\vF_q$}\}.\]
Notice that $k_f$ in principle might be non trivial (an example is $k_f$ for $f=x^3$ and $q\equiv 2\mod 3$: in this case $M_f=k_f(t)[x]/(f-t)$ and $k_f=\vF_{q^2}$).
Recall that $\vF_q(x)\cong \vF_q(t)[x]/(f(x)-t)$. Let $G_f$ be the \emph{monodromy group} of $f$, i.e. the Galois group of the the  field extension $M_f:\vF_q(t)$. 
When we refer to the \emph{genus} $g_f$ of $M_f$ we consider $M_f$ as a function field over its field of constants $k_f$. In all interesting cases we will have $k_f=\vF_q$ so this distinction will not affect us.
For an element $g\in G_f$ and a place $P\in \mathcal P^1(\vF_q(t))$ we say that $g$ is a Frobenius for $P$ if there exists a place $R$ lying above $P$ such that $g(R)=R$ and the map $g_R:\cO_R/R\rightarrow \cO_R/R$ acts as $g(y)=y^q$. In particular, we say that the identity is a Frobenius for $P$ if $\cO_R/R\cong \cO_P/P\cong\vF_q$.

In the rest of the paper we will identify the places of degree $1$  of $\vF_q(t)$ with $\vF_q\cup\{\infty\}$. For a finite set $A$, we denote its cardinality by $\# A$. Let us denote the symmetric group of degree $m$ by $\mathcal S_m$ and the alternating group by $\mathcal A_m$.
Let us denote the multiplicative group $\vF_q\setminus \{0\}$ by $\vF_q^*$.
We say that a polynomial $f\in\vF_q[X]$ is totally split if it factors as a product of $\deg(f)$ distinct linear factors.
If $F_1$ and $F_2$ are fields contained in a larger field $L$, we denote by $F_1F_2$ their compositum, i.e. the smallest subfield of $L$ containing $F_1$ and $F_2$.

\section{Locally Recoverable Codes and Good Polynomials}
Let us start with the fundamental definition which is slightly refined compared to \cite{tamo2014family,liu2018new}, as we want to keep track of the size of the covering attached to the good polynomial, as it relates directly to the maximal achievable dimension and the length of the codes constructed with the given good polynomial.
\begin{definition}
Let $f\in \vF_q[X]$ and let $\ell$ be a positive integer. Then $f$ is said to be $(r,\ell)$-good if 
\begin{itemize}
\item $f$ has degree $r+1$,
\item there exist  $A_1,\dots A_\ell$ distinct subsets of $\vF_q$ such that
\begin{itemize}
\item for any $i \in \{1,\dots \ell\}$, $f(A_i)=\{t_i\}$ for some $t_i\in \vF_q$, i.e. $f$ is constant on $A_i$,
\item $\# A_i= r+1$,
\item $A_i\cap A_j=\emptyset$ for any $i\neq j$.
\end{itemize}
\end{itemize}
We say that the family $\{A_1,\dots A_\ell\}$ is a \emph{splitting covering} for $f$. We say that a polynomial is $r$-good if it has degree $r+1$ and it is $(r,\ell)$-good for some $\ell>0$. For simplicity of notation, we allow $\ell$ to be negative or zero, in which case an $(r,\ell)$-good polynomial is simply a polynomial of degree $r+1$.  A polynomial that is \emph{not good} is a polynomial such that there is no $t_0\in \vF_q$ such that $f(X)-t_0$ splits completely in $\deg(f)$ distinct linear factors.
\end{definition}

\begin{remark}\label{rem:monomial}
Observe that if a polynomial of degree $r+1$ is  $(r,\ell)$-good, then  $\ell$ is  at most $\lfloor q/(r+1) \rfloor$.
\end{remark}

\begin{remark}
Notice that an $(r,\ell)$-good polynomial is also $(r,t)$-good for any $t\leq \ell$, as one can simply drop some of the $A_i$'s.
\end{remark}

Let us recall the definition of optimal LRCs \cite{tamo2014family}.
\begin{definition}
A $(n,k,r)$-LRC $\mathcal C$ is said to be \emph{optimal} if the minimum distance $d$ of $\mathcal C$ satisfies 
\[d= n-k-\left\lceil\frac{k}{r}\right\rceil+2.\]
\end{definition}
In fact it can be proven that $n-k-\left\lceil\frac{k}{r}\right\rceil+2$ is the maximum distance achievable by any $(n,k,r)$-LRC \cite[Theorem 1]{papailiopoulos2014locally}.

The following result is \cite[Construction 1]{tamo2014family}. We write it in the format of \cite[Theorem 4]{liu2018new}, which is more convenient for our purposes.

\begin{theorem}\label{thm:LRCconstruction}
Let $r\geq 1$ be a positive integer and $g$ be an  $(r,\ell)$-good polynomial over $\vF_q$ for the set $A=\bigcup^\ell_{i=1}A_i$. Let $t\leq \ell$, $n=(r+1)\ell$ and $k=rt$.
For $a=(a_{i,j}, i=0,\dots r-1;\, j=0,\dots t-1)\in \vF_q^k$, let
\[f_a(X)=\sum_{i=0}^{r-1}\sum_{j=0}^{t-1}a_{i,j}g(X)^jx^i.\]
Define \[\mathcal C=\{(f_a(x),\, x\in A)\,|\, a\in \vF_q^k\}.\]
Then $\mathcal C$ is an optimal $(n,k,r)$-LRC over $\vF_q$.
\end{theorem}

The following is the key observation.
\begin{remark}\label{rem:fundamental_remark}
Let $L_f=\vF_q(x)$, where $x$ is any root of $f-t$ in $M_f$.
It is easy to see that each of the $A_i$'s corresponds to a totally split place of degree $1$ of the extension $L_f:\vF_q(t)$: if $f(A_i)=\{t_i\}$ for some $t_i\in \vF_q$, then the place corresponding to $t_i$ factors as a product of exactly $r+1$ places of $L_f$, which themselves correspond to the elements of $A_i$. 

Clearly, the correspondence between the $A_i$'s and the totally split places of $\vF_q(t)$ is injective and simply given by $A_i\mapsto f(A_i)=\{t_i\}$.
 In addition, the maximum $\ell$ for which $f$ is $(r,\ell)$-good is the number of totally split places in $\vF_q(t)$ of the extension $\vF_q(x):\vF_q(t)$.
As Theorem \ref{thm:LRCconstruction} shows, a large $\ell$ is desirable as for fixed locality $r$ it allows constructions of optimal codes with parameters $((r+1)\ell, rt, r)$.
\end{remark}

\begin{remark}\label{rem:r1poly}
Notice that it is obvious to construct an $(r,1)$-good polynomial as it is enough to take a totally split polynomial. That allows to construct only one LRC  with parameters $(r+1,r,r)$. This is the reason why in this paper we include both $r$ and $\ell$ in the notion of ``good'' polynomial.
\end{remark}

\section{Galois Theory over Global Function Fields and Good Polynomials}\label{sec:GalThFF}

We now briefly explain the essence of the proposed method. We start with a polynomial $f$ and we are interested in the number of $t_0$'s in $\vF_q$ such that $f-t_0$ splits completely, by Remark \ref{rem:fundamental_remark}. Let $t$ be trascendental over $\vF_q$ and let us consider the extension $\vF_q(x):\vF_q(t)$ where $x$ is a root of $f-t$ over the algebraic closure of $\vF_q(t)$. The splitting of the places of degree $1$ of the extension $\vF_q(x):\vF_q(t)$ corresponds to the factorization shapes of $f-t_0$, when $t_0$ varies in $\vF_q$. Unfortunately, for most polynomials $f$, the extension $\vF_q(x):\vF_q(t)$ is not Galois, but if we take the Galois closure $M_f$ of such extension, we can still extract information about the splitting of the places in $\vF_q(x):\vF_q(t)$ by looking at the disjoint cycle decomposition of the elements of the Galois group (this is a classical fact, but see for example  \cite[Lemma 1]{ferraguti2018complete} or \cite[Theorem 2.3]{micheli2019selection}): as we are interested in the totally split places, we consider the only element with all fixed points, the identity.
As long as the field of constants $k_f$ of $M_f$ is simply $\vF_q$, Chebotarev Density Theorem for function fields applied on the identity (see for example \cite{kosters2017short}) ensures that the number of totally split places is accurately estimated by \cite[Corollary 1]{kosters2017short}, up to an error term of $O(\sqrt q)$ where the implied constant is independent of $q$.

The following proposition summarises all the results we need from algebraic number theory in a compact way. We include the proof for completeness.
\begin{proposition}\label{thm:estimate_totsplit}
Let $f\in \vF_q[X]$ be a separable polynomial.
Let $g_f$ be the genus of $M_f$ and
let $x$ be a root of $f(X)-t$ in $M_f$.
\begin{itemize}
\item[(i)] If the extension field $\vF_q(x):\vF_q(t)$ has a totally split place $P$ of degree $1$, then $k_f=\vF_q$.
\item[(ii)] Suppose that $k_f=\vF_q$. Let $T^1_\text{split}(f)$ be the set of $t_0\in \vF_q$ such that $f-t_0$ factors in $\deg(f)$ distinct linear factors over $\vF_q$. Then
\[\frac{q+1-2g_f\sqrt{q}}{\# G_f}-\frac{\#\Ram^1(M_f:\vF_q(t))}{2}\leq \#T^1_\text{split}(f) \leq \frac{(q+1)+2g_f\sqrt{q}}{\# G_f},
\]
where $\Ram^1(M_f:\vF_q(t))$ is the set of ramified places of degree $1$ of the extension $M_f:\vF_q(t)$.
\item[(iii)] Let $C(f)$ be the smallest integer such that 
\[\frac{C(f)+1-2g_f\sqrt{C(f)}}{\# G_f}-\frac{\#\Ram^1(M_f:\vF_q(t))}{2}>0.
\]
 If $q>C(f)$, and $k_f=\vF_q$, then $\vF_q(x):\vF_q(t)$ has a totally split place.
\end{itemize}
\end{proposition}
\begin{proof}
Let us prove (i).  Since $M_f$ is the Galois closure of $\vF_q(x):\vF_q(t)$, by \cite[Lemma 3.9.5.]{stichtenoth2009algebraic} $P$ is totally split also in $M_f:\vF_q(t)$.
Let $R\subseteq M_f$ be a place lying above the totally split place $P\in \mathcal P^1(\vF_q(t))$.  Since $P$ is totally split and of degree $1$ we have that $[\cO_R/R:\cO_P/P]=1$ and then $\cO_R/R\cong \vF_q$. By \cite[Proposition 1.1.5, (c)]{stichtenoth2009algebraic} we have that $\cO_R/R$ contains the field of constants of $M_f$, and in turn $k_f$ cannot be a proper extension of $\vF_q$.

Let us now prove (ii). Let $P$ be a place of $\vF_q(t)$. Since $M_f$ is a Galois extension of $\vF_q(t)$ all places of $M_f$ lying above $P$ have the same relative degree $f(P)$ and ramification index $e(P)$ by \cite[Corollary 3.7.2]{stichtenoth2009algebraic}. 
\vspace{0.5cm}

\textbf{Claim 1.} \emph{A place $P$ in $\vF_q(t)$ is totally split in $M_f:\vF_q(t)$ if and only if it is unramified and the identity is a Frobenius for $P$.}
\begin{proof}[Proof of claim 1]
A place $P$ is totally split if and only if any place $R\subseteq M_f$ lying over $P$ is unramified and $f(P)=1$, which happens if and only if $P$ is unramified and the identity of $G_f$ acts as a Frobenius for the (trivial) field extension $\cO_R/R:\cO_P/P$.
\end{proof}

\textbf{Claim 2.} \emph{Let $t_0\in \vF_q$. The polynomial $f(X)-t_0\in  \vF_q[X]$ splits into $\deg(f)$ distinct linear factors if and only if the place $P_0$ corresponding to $t_0$ in $\vF_q(t)$ is totally split in $M_f:\vF_q(t)$.}
\begin{proof}[Proof of claim 2]
Recall that $x\in M_f$ is a root of $f(X)-t$. First observe that the splitting of degree $1$ places in the extension $\vF_q(x):\vF_q(t)$ correspond exactly to the factorization of $f(x)-t_0$, for $t_0\in \vF_q$. Since $M_f$ is the Galois closure of the extension $\vF_q(x):\vF_q(t)$, then \cite[Lemma 3.9.5.]{stichtenoth2009algebraic} ensures that $P_0$ is also totally split in $M_f:\vF_q(t)$. Viceversa, since ramification and relative degrees are multiplicative in intermediate extensions, it is easy to see that if $P_0$ is totally split in the extension $M_f:\vF_q(t)$, then it is also totally split in $\vF_q(x):\vF_q(t)$.
\end{proof}

Using the claims above and observing that the place at infinity of $\vF_q(t)$ is ramified, we reduced the problem to finding all places $P\subset \vF_q(t)$ that are totally split in $M_f:\vF_q(t)$.

We want to use \cite[Corollary 1]{kosters2017short}. Since in our case $k_f=\vF_q$  the field of constants extension is trivial. In the notation of Koster's corollary we have that $M=M_f$, $N=G_f$, and $\Id\in \overline F$. Moreover, we are interested in $\gamma=\Id$, therefore $(P,M_f)(\Id)= \delta_P/e(P)$, where $\delta_P=1$ if the identity is a Froebenius at $P$ (i.e. $P$ splits into factors of relative degree $1$ and multiplicity $e(P)$), and $0$ otherwise. Then we have

\[
\left| \sum_{P\in \mathcal P^1(\vF_q(t)/\vF_q)} \frac{\delta_P}{e(P)}- \frac{1}{\# G_f} (q+1)\right| \leq \frac{2}{\# G_f}g_f \sqrt q,
\]

We now want to split the sum for ramified and unramified places so we set 
\[\mathcal R=\sum_{\substack{P\in \Ram^1(M_f:\vF_q(t))\\ \text{ $\Id$ is a Frobenius at $P$} }} \frac{1}{e(P)}, \quad \text{and} \quad \mathcal U=\left(\sum_{\substack{\text{unramified } P\in \mathcal P^1(\vF_q(t)/\vF_q): \\ \text{ $\Id$ is a Frobenius at $P$}}} 1\right)\]
so that we get

\[\left| \mathcal{R} + 
\mathcal U- \frac{1}{\#  G_f} (q+1)\right|\leq \frac{2}{\# G_f}g_f \sqrt q.\]

Observing now that  $e(P)\geq 2$, the place at infinity is ramified, and
\[\# T^1_\text{split}(f)=\mathcal U,\]
the final claim follows directly.

The statement in (iii) is immediate by observing that the condition on $q$ ensures the existence of a totally split place by point (ii).
\end{proof}

Point (ii) of the proposition above essentially states a very classical fact from algebraic number theory: since the number of ramified places is bounded by an absolute constant (depending only on the degree of $f$) the set of  totally split places of a Galois extension of global fields has density $1/\# G$ where $G$ is the Galois group of the extension field. When the global fields are actually global function fields (i.e. finite extensions of $\vF_q(t)$) the number of degree one totally split places can be estimated as in \cite[Corollary 1]{kosters2017short}, leading to the estimate in (ii). The estimate is essentially optimal, as the Riemann Hypothesis for curves over finite fields is proved and is equivalent to the Hasse-Weil bound.

We notice now how $\# \Ram^1(M_f:\vF_q(t))$ and $g_f$ can be explicitely bounded by a constant depending on the degree of $f$ and independent of $q$.

\begin{remark}\label{rem:genus_zero}
Proposition \ref{thm:estimate_totsplit}, point (ii) is in the format of an estimate, but whenever $g_f$ is zero, it can be used to obtain an exact formula, if in the proof one works out exactly what happens at the ramified places for the quatities $\delta_P/e(P)$.
\end{remark}

\begin{proposition}\label{thm:estimates_g_and_R}
Let $f\in \vF_q[X]$ be a separable polynomial, $n=\deg(f)$, and $g_f$ be the genus of the splitting field $M_f$ of $f-t$ over $\vF_q(t)$. Let $\Ram(M_f:\vF_q(t))$ (resp. $\Ram^1(M_f:\vF_q(t))$) be the set of ramified places of $\vF_q(t)$ (resp. ramified places of degree $1$) in the field extension $M_f:\vF_q(t)$. Then
\begin{enumerate}
\item[(i)] $\# \Ram^1(M_f:\vF_q(t))\leq \# \Ram(M_f:\vF_q(t))\leq n$.
\item[(ii)] Suppose that $\Char(\vF_q)=
p\nmid \# G_f$ 
and $k_f=\vF_q$. Then $g_f\leq ((n-2)\# G_f+2)/2$.
\end{enumerate}
\end{proposition}
The proof of the above statement is immediate by observing that (i) follows from the fact that a place $P$ of $\vF_q(t)$ is ramified in the extension $\vF_q(x):\vF_q(t)$ if and only if it is ramified in its Galois closure $M_f:\vF_q(t)$, while (ii) is just an application of Hurwitz genus formula for Galois extensions.

\subsection{Existence of good polynomials}\label{sec:existencegood}
In this section we prove some existential results over base fields which are relatively large compared with the locality parameter $r$.
\begin{proposition}\label{thm:generic_case}
Let $r$ be a positive integer and $q$ be a prime power. Then any separable polynomial of degree $r+1$ over $\vF_q$ such that $k_f=\vF_q$ is at least $(r,\ell)$-good, with $\ell$ at least
\[\frac{(q+1)-2g_f\sqrt{q}}{(r+1)!}-\deg(f)/2.\]
Moreover, if $k_f\neq \vF_q$ then $f$ is not a good polynomial.
\end{proposition}
\begin{proof}
First observe that $G_f$ can be seen as a subgroup of the symmetric group $\mathcal S_{r+1}$, which forces $\# G_f\leq (r+1)!$. 
The statement follows immediately by applying point (ii) of Proposition \ref{thm:estimate_totsplit} and bounding $\# \Ram^1(M_f:\vF_q(t))$ with point (i) of Proposition \ref{thm:estimates_g_and_R}.
If $k_f\neq \vF_q$, then the statement (i) of Proposition \ref{thm:estimate_totsplit} applies, so $f$ cannot have a totally split place and therefore it cannot be $(r+1,\ell)$-good for any positive integer $\ell$.

\end{proof}

\begin{remark}
The condition $k_f=\vF_q$ is generic. In fact, the conditions given for example in \cite[Theorem 3.6]{turnwald} (for $K=\overline\vF_q$) for $G_f$ to be symmetric over $\overline{\vF}_q(t)$ are all open. This forces the geometric Galois group to be symmetric most of the times, which force the arithmetic Galois group to be symmetric as well, and therefore the field of constants extension to be trivial. Actually we will see in practical examples that it is very easy to force $k_f=\vF_q$, by simply using point (i) of Proposition \ref{thm:estimate_totsplit}. 
\end{remark}

\begin{proposition}\label{thm:pessimistic}
Let $B\subseteq \vF_q$ be a set of size $r+1$.
Let $r$ be a positive integer such that $\gcd(q,(r+1)!)=1$ and $f=\prod_{b\in B}{(X-b)}$. Then $f$ is  $(r,\ell)$-good, with $\ell$ at least
\[\frac{q+1}{(r+1)!}-\left(r-1+\frac{2}{(r+1)!}\right)\sqrt{q}-\frac{r+1}{2}.\]
\end{proposition}
\begin{proof}
Clearly $f$ is separable because the characteristic does not divide degree of the polynomial.
First using (i) of Proposition \ref{thm:estimate_totsplit} we get that $k_f=\vF_q$, as $t=0$ is a totally split place. Now using Proposition \ref{thm:generic_case} we get that $f$ is at least $(r,\frac{(q+1)-2g_f\sqrt{q}}{(r+1)!}-\deg(f)/2)$-good. Using now the bound on $g_f$ given by (ii) of Proposition \ref{thm:estimates_g_and_R} we get the wanted result.
\end{proof}

\begin{remark}
The proposition above implies that forcing just one totally split place immediately gives the existence of many totally split places, which is an interesting fact.
It is worth noticing here that the worst case scenario given in Proposition \ref{thm:pessimistic} is actually the generic case: if one fixes a random polynomial of degree $n$ in $\vF_q[X]$ and considers $G_f$, then most likely $G_f=\mathcal S_n$, in which case the estimate for the number of totally split places is given by $q/(n!)$. Table \ref{table_of_x5_smt} shows some instances of this fact for degree $5$.
\end{remark}

\begin{remark}
 The reader should observe that Proposition \ref{thm:pessimistic} and \cite[Proposition 3.4]{tamo2014family} have the same asymptotic in $q$ which is simply $q/(r+1)!$. The most important difference is that Proposition \ref{thm:pessimistic} is constructive as it ensures that \emph{however} we fix a polynomial of the form $f=\prod_{b\in B}{(x-b)}$, this has splitting covering of size roughly $q/(r+1)!$.
\end{remark}

\subsection{Selection of very good polynomials}\label{sec:selection}

The method provides existential results and fits the existing literature on good polynomials in a Galois theoretical context, but also allows to produce new good polynomials, which is the most important application. 
We give emphasize here  that we want ``very good'' polynomials, i.e. $(r,\ell)$-good polynomials such that $\ell$ is as large as possible (we already noticed in Remark \ref{rem:r1poly} that building an $r$-good polynomial is a trivial task if $\ell$ is not taken into account, but then length and dimension of the code are small). 
Let us start with the two fundamental and well known constructions from \cite{liu2018new,tamo2014family}.

\begin{proposition}\label{thm:previous_constructions}
The following are $r$-good polynomials
\begin{itemize}
\item Let $r+1=m$ be a divisor of $q-1$. The polynomial $X^{m}$ is $(r,(q-1)/m)$-good.
\item Let $V$ be an additive subgroup of $\vF_q$. The polynomial  
$\prod_{v\in V} (X-v)$ is $(\#V-1, q/\#V)$-good. 
\end{itemize}
\end{proposition} 
If one combines the construction above via composition, one can get new good polynomials as described in \cite[Theorem 3.3]{tamo2014family} and \cite[Section A,B]{liu2018new}.

\begin{remark}
All the good polynomials above have in common that $k_f=\vF_q$ and $g_f=0$, so thanks to Remarks \ref{rem:genus_zero} and \ref{rem:fundamental_remark}, they fit completely in the easy case (the genus zero condition) of our context. In particular, $k_f=\vF_q$ (which is a necessary condition for a polynomial to be good) exactly when $m\mid q-1$ or the linearised polynomial $\prod_{v\in V} (X-v)$  splits completely over $\vF_q$. 

Also, the splitting field of their composition inherits nice properties so the results in \cite[Theorem 3.3]{tamo2014family} and \cite[Section A,B]{liu2018new} can also be derived from our framework.

For the sake of explanation of our method, let us fit for example the case of $f=X^m\in \vF_q[X]$ in our context. For simplicity, let us assume $m\geq 3$. First, observe that such example of good polynomial exists only when $m\mid q-1$, which is exactly the condition needed to have $k_f=\vF_q$: in fact, if $m\mid q-1$ then $\vF_q$ contains a primitive $m$-th root of unity and therefore we have that $M_f\cong \vF_q(t)[x]/(f(x)-t)$ as it is enough to add just one root of $f-t$ to $\vF_q(t)$ to get all the other roots.
Now, $\vF_q(x)=\vF_q(t)[x]/(f(x)-t)=M_f$ as $t$ is just another name for $f(x)$ so the function field $\vF_q(x)$ is still rational, so $g_f$ is zero. Trivially 
$\# G_f=[M_f:\vF_q(t)]=m$. In addition
$\Ram^1(M_f: \vF_q(t))=\{\infty, 0\}$ so we get 
\[\frac{q+1}{m}-1\leq \#T^1_\text{split}(f)\leq \frac{q+1}{m}.\]
Since $\#T^1_\text{split}(f)$ has to be an integer and $m\geq 3$, we obtain directly that $\ell=\#T^1_\text{split}(f)=(q-1)/m$.
In this case the direct proof is actually easier than the one proposed here, on the other hand it only works thanks to the cyclic group structure of $\vF_{q}^*$ while our approach works in general, as it is based on a density argument.
\end{remark}

\textbf{The method.} Constructing good polynomials  of given degree using Proposition \ref{thm:estimate_totsplit} is very simple: the number of $t_0$'s such that $f-t_0$ splits into $\deg(f)$ distinct linear factors can be estimated with $q/\#G_f$ (as long as $k_f=\vF_q$, otherwise it is the empty set) up to an error term depending on  $g_f\sqrt{q}/ \# G_f$ and on the ramified places of degree $1$.
Informally, to construct an $(r,\ell)$-good polynomial the task becomes the following: find a class of polynomals $f$ of degree $r+1$ such that the Galois groups of $f-t$ over $\vF_q(t)$ include in some transitive subgroup $G$ of $\cS_{r+1}$ and such that $q/\ell$ is roughly of the size of $G$. Now sieve the family allowing only polynomials with minimal ramification over the base field. 
You do not need to compute exactly the Galois group, but only the fact that $G_f$ includes in a small group is enough. There is a remarkable class of examples arising from arithmetic dynamics whenever $r+1$ is a composite number: in this case one can systematically compose non trivial polynomials over $\vF_q$ to obtain a polynomial $f$ of degree $r+1$. Now $G_f$ will be significantly smaller than $\cS_{r+1}$ as it is contained in the automorphism group of a certain tree (not necessarily binary) of height equal to the number of prime factors of $r+1$.

We give now simple applications of the method in the case $r=2$ with $\#G_f=6$, in the case $r=3$ with $\#G_f=8$, and $r=6$ with $\# G_f=12$.
We want to stress here that these constructions hold for almost all $q$'s.

\begin{theorem}\label{thm:degree_3_pred}
Let $q$ be a prime power, $b\in \vF_q$ and $f=x(x-1)(x-b)\in \vF_q[X]$.
Then
\begin{itemize} 
\item[(i)] if $q$ is even then $f$ is $(2,\ell)$-good with $\ell$ at least
\[\left\lceil{\frac{q+1-2\sqrt{q}}{6}-1}\right\rceil,\]
\item[(ii)] if $q=3^m$, then $f$ is $(2,\ell)$-good with $\ell$ at least 
\[\left\lceil{\frac{q+1-2\sqrt{q}}{6}-1} \right\rceil,\]
Moreover, if $b=-1$ then $f$ is $(2,\ell)$-good with $\ell$ at least 
\[\left\lceil{\frac{q+1-2\sqrt{q}}{6}-\frac{1}{2}}\right\rceil,\]
\item[(iii)] if $q$ is odd such that $q\mod 3 \neq 0$, and $1-b+b^2$ is not a square in $\vF_q$, then $f$ is $(2,\ell)$-good with $\ell$ at least \[\left\lceil{\frac{q+1-2\sqrt{q}}{6}-\frac{1}{2}}\right\rceil.\]

\end{itemize}
\end{theorem}
\begin{proof}
Applying (i) of Proposition \ref{thm:estimate_totsplit} gives that $k_f=\vF_q$ and therefore we can apply (ii). Clearly $\# G_f$ is at most $6$, as it has to be a subgroup of the symmetric group. $M_f$ is constructed by simply adding two roots $\{x_1,x_2\}$ of $f-t$, since the third one can always be obtained from the other two with field operations. Therefore $M_f=\vF_q(t)(x_1,x_2)=\vF_q(x_1,x_2)$ ($t$ can be obtained by evaluating $f$ at $x_1$ for example), and then $g_f\leq 1$ by Riemann's inequality \cite[Corollary 3.11.4]{stichtenoth2009algebraic} (the minimal polynomial of $x_1$ over $\vF_q(x_2)$ has degree at most $2$ and viceversa).
The ramified places at finite of $M_f$ are in correspondence with the zeroes of the derivative  of $f$. As usual, the place at infinity is always ramified.
\begin{itemize}
\item For even $q=2^m$, $f'=x^2+b=(x+b^{2^{m-1}})^2$ so that $\# \Ram^1(M_f:\vF_q(t))=2$.
\item For $q=3^m$, $f'=-2(b+1)x+b=(b+1)x+b$, which has either no roots ($b+1=0$ and $\# \Ram^1(M_f:\vF_q(t))=1$) or one root ($b+1\neq 0$ and $\# \Ram^1(M_f:\vF_q(t))=2$).
\item For odd $q$, $q\mod 3 \neq 0$ we have $f'=3 x^2-2(b+1)x+b$, which has no roots exactly when $b^2-b+1$ is not a square.
\end{itemize}
Using the formula in (ii) of Proposition \ref{thm:estimate_totsplit} we get the wanted results.

\end{proof}

\begin{remark}
Notice that the construction is not exploiting the multiplicative nor the additive subgroups of $\vF_q$. With this method we are able for example to write a polynomial of fixed degree (see for example the case of Theorem \ref{thm:degree_3_pred}) in $\vZ[X]$ that is $(r,\ell(q))$-good at any $q$, with $\ell(q)$ being an explicit constant.
\end{remark}

Suppose now that we want to construct a $(n,k,3)$-code over an alphabet of size $q$, with $q\equiv 3 \mod 4$. For some $\ell$ we would need an $(3,\ell)$-good polynomial $f$. None of the constructions in \cite{liu2018new} apply as we now explain. First of all, observe that $f$ cannot be  a composition of a non-trivial linearised polynomial and a power function, as its degree is $4$ and  $q$ is odd.
If $f$ was a power function, then (up to multiplication by a scalar) $f=x^4$, but then $4\nmid q-1$ and so $x^4-t_0$ is never totally split for any $t_0\in \vF_q$.
The following result provides a $(r,\ell)$-good polynomial with $\ell$ roughly of the size of $\lfloor{(q-1)/8}\rfloor$.

\begin{theorem}\label{thm:x4x2prediction}
Let $q\geq 5$ be an odd prime power. Let $a\in \vF_q^*$ and  $f=X^4+aX^2\in \vF_q[X]$. Then $\# G_f=8$ and
\begin{itemize}
\item if $-a/2$ is not a square in $\vF_q$, the polynomial $f$ is  $(3,\ell)$-good with $\ell$ at least 
\[\left\lceil{\frac{q+1}{8}-1}\right\rceil,\]
\item if $-a/2$ is a square in $\vF_q$, the polynomial $f$ is  $(3,\ell)$-good with $\ell$ at least 
\[\left\lceil{\frac{q+1}{8}-2}\right\rceil.\]
\end{itemize}
\end{theorem}
\begin{proof}
First we compute the splitting field of $f-t$. It is clear that 
\[M_f=\vF_q\left(t,\sqrt{\frac{-a+\sqrt{a^2+4t}}{2}},\sqrt{\frac{-a-\sqrt{a^2+4t}}{2}}\right).\] 

Since $\sqrt{\frac{-a+\sqrt{a^2+4t}}{2}}\sqrt{\frac{-a-\sqrt{a^2+4t}}{2}}=\sqrt{-t}$, then we have that 
\[M_f=\vF_q\left(t,\sqrt{-t},\sqrt{\frac{-a+\sqrt{a^2+4t}}{2}}\right).\] By the tower law we have that $\#G_f=[M_f:\vF_q(t)]\leq 8$.
In order to show that $[M_f:\vF_q(t)]=8$ it is enough to show that 
$\sqrt{-t}\notin \vF_q\left(t,\sqrt{\frac{-a+\sqrt{a^2+4t}}{2}}\right)$. By contradiction, suppose that $\sqrt{-t}\in \vF_q\left(t,\sqrt{\frac{-a+\sqrt{a^2+4t}}{2}}\right)$. This can happen only when $\vF_q\left(t,\sqrt{\frac{-a+\sqrt{a^2+4t}}{2}}\right)$ is Galois over $\vF_q(t)$, and therefore $[M_f:\vF_q(t)]=4$. 
Consider the subfields 
$F_1=\vF_q\left(t,\sqrt{a^2+4t}\right)$ and $F_2=\vF_q\left(t,\sqrt{-t}\right)$. $F_1$ and $F_2$ are distinct and satisfy $[F_1:\vF_q(t)]=[F_2:\vF_q(t)]=2$. 
By contradiction, let us assume $[M_f:\vF_q(t)]=4$. One observes that, since  $F_1\cap F_2=\vF_q(t)$ then $[F_1F_2:\vF_q(t)]=4$, and therefore $M_f$ can be written as a compositum of $F_1$ and $F_2$, i.e. $M_f=F_1F_2=\vF_q\left(t,\sqrt{a^2+4t},\sqrt{-t}\right)$. But one can show directly that $\frac{-a+\sqrt{a^2+4t}}{2}$ is not a square in $F_1F_2$, so $F_1F_2$ does not contain all the roots of $f-t$ and therefore cannot be its splitting field.

We now compute $g_f$ and $k_f$.
For simplicity let us denote $x=\sqrt{\frac{-a+\sqrt{a^2+4t}}{2}}$ and $y=\sqrt{\frac{-a-\sqrt{a^2+4t}}{2}}$. It is clear that $M_f=\vF_q(t,x,y)$.
Also, $\vF_q(t,x,y)=\vF_q(x,y)$ since $t$ can be obtained using only $x$ and elementary $\vF_q$-operations that do not involve $t$. The elements $x$ and $y$ verify the equation $x^2+y^2+a=0$, which is a conic so its genus is zero. 
Another way to see the fact that $g_f=0$ is to use the inequality in \cite[Proposition 3.11.5]{stichtenoth2009algebraic} on the equation $x^2+y^2+a=0$.
Moreover by \cite[Corollary 3.6.8]{stichtenoth2009algebraic}, $k_f=\vF_q$ because $x^2+y^2+a$ is absolutely irreducible if $a\neq 0$.

Let us now compute $\# \Ram^1(M_f:\vF_q(t))$. Since the ramified places of degree $1$ correspond to the zeroes in $\vF_q$ of the derivative of $f$ (the place at infinity is always ramified), it is easy to see that $\Ram^1(M_f:\vF_q(t))=\{\infty,0\}$ if $-a/2$ is not a square in $\vF_q$ and $\Ram^1(M_f:\vF_q(t))=\{\infty,0,b,-b\}$ otherwise, where $b\in \vF_q$ is an element such that $b^2=-a/2$.
A direct application of Proposition \ref{thm:estimate_totsplit}  gives now the wanted result.

\end{proof}

Let us finish with an example of a $(5,\ell)$-good polynomial. Again, let us explain why this is a new example of a good polynomial. Fix the size of the base field to be $q$ and assume that one wants to construct an LRC with locality $5$.
Suppose that $6\nmid q-1$ and $q=p^m$ is not divisible by $2$ or $3$. Then one would need a degree $6$ polynomial such that $f-t_0$ is totally split for many $t_0$'s in $\vF_q$.
But then, none of the constructions in Proposition \ref{thm:previous_constructions} will work, nor compositions of those: in fact $f$ cannot be a composition (possibly trivial) of a $p$-linearised polynomial with a power function for degree reasons, and also cannot be a power function because $6\nmid q-1$, and therefore $x^6$ is not a good polynomial over $\vF_q$.

\begin{theorem}\label{thm:degree_6_pred}
Let $q$ be an odd prime power such that $q\not\equiv 0 \mod 2,3$ and $a\in \vF_q^*$ such that $a$ in not a square. Let $f=(X^3-aX)^2\in \vF_q[X]$.
Then $\#G_f=12$ and $f$ is a $(5,\ell)$-good polynomial with $\ell$ at least
\[\left\lceil{\frac{q+1-2\sqrt{q}}{12}-2}\right\rceil.\]
\end{theorem}
\begin{proof}
First, we need to compute $M_f$, the splitting field of $f-t$. Observe that $\sqrt{t}\in M_f$ and that adding $\sqrt{t}$ to $\vF_q(t)$ allows the splitting
\[f-t=(X^3-aX-\sqrt{t})(X^3-aX+\sqrt{t}).\]
Set $H_1=X^3-aX-\sqrt{t}$ and $H_2=X^3-aX+\sqrt{t}$.
We need now to split $H_1$ over $\vF_q(\sqrt{t})$, because in that case also $H_2$ splits as $H_2(X)=-H_1(-X)$.
Since $H_1(X)$ is irreducible over $\vF_q(\sqrt{t})$, then $N=\Gal(M_f:\vF_q(\sqrt{t}))$ is equal to $\cS_3$ or $\mathcal A_3$. But since one can check directly that the discriminant $\Delta=4a^3-27t$ of $H_1$ over $\vF_q(\sqrt{t})$ is not a square in $\vF_q(\sqrt{t})$, then $N\neq \mathcal A_3$, which forces $N=\cS_3$ and therefore $[M_f:\vF_q(\sqrt{t})]=\# N=6$. By the tower law we have 
\[\#G_f=[M_f:\vF_q(t)]= [M_f:\vF_q(\sqrt{t})]\cdot [\vF_q(\sqrt{t}):\vF_q(t)]=12.
\]
We now want to show that $k_f=\vF_q$ so that we can apply point (ii) of Proposition \ref{thm:estimate_totsplit}. But this is completely obvious because one can consider $\Gal(M_f:k_f(t))$, as $k_f(t)$ is clearly a subfield of $M_f$, and exactly the same arguments as above apply. This shows that 
\[\#\Gal(M_f:k_f(t))=[M_f:k_f(t)]= [M_f:k_f(\sqrt{t})]\cdot [k_f(\sqrt{t}):k_f(t)]=12.
\]
Since $\Gal(M_f:k_f(t))\subseteq G_f$ we must have equality, and therefore $k_f(t)=\vF_q(t)$, which in turn forces $k_f=\vF_q$.

Let $y$ and $z$ be two roots of $H_1$. Observe that $M_f=\vF_q(t,\sqrt{t},y,z)=\vF_q(y,z)$ as $\sqrt{t}$ (and so $t$) can be obtained by evaluating $ X^3-aX$ at $y$ for example, and adding two roots of $H_1$ is enough to obtain the third root using field operations. Then we have $g_f\leq 1$ by Riemann's inequality \cite[Corollary 3.11.4]{stichtenoth2009algebraic} (the minimal polynomial of $y$ over $\vF_q(z)$ has degree at most $2$ and viceversa).

Let us now compute $\Ram^1(M_f:\vF_q(t))$. As usual, we look at the number of zeroes in $\vF_q$ of the derivative $f'(X)=2(X^3-aX)(3X^2-a)$. Since $a$ is not a square, $f'(X)$ has at most $3$ zeroes in $\vF_q$, depending on whether $3$ is a square or not in $\vF_q$. Since the place at infinity is always ramified we have that $\Ram^1(M_f:\vF_q(t))$ consists of at most $4$ places.

Plugging everything in the formula off (ii) of Proposition \ref{thm:estimate_totsplit} we get the wanted result.

\end{proof}

\begin{remark}
To further limit the size of the ramified places, in Theorem \ref{thm:degree_6_pred} we could also impose that $a/3$ is not a square (but this would restrict the fields where the theorem holds to the ones where $3$ is a square).
\end{remark}

\begin{remark}
As we already observed just before the statement of Theorem \ref{thm:degree_6_pred}, for infinitely many base fields $\vF_q$ none of the specific constructions using monomial and linearised polynomial would work. Also, \cite[Proposition 3.4]{tamo2014family} would only give an $(5,\ell)$-good polynomial with $\ell=\lceil {q\choose 6}/q^5 \rceil\leq \lceil q/720 \rceil$. With our construction we obtain polynomials with a splitting covering of size $\ell$ asymptotic to $q/12$, in turn obtaining optimal codes with same locality $5$ but with dimension and length roughly $60$ times larger (of course, if one wanted one can always drop some elements of the splitting covering to obtain a smaller code). Notice that $60$ is the index of $G_{(x^3-ax)^2}$ in $S_{6}$, which in fact relates directly to the dimension boost compared with \cite[Proposition 3.4]{tamo2014family}.
\end{remark}

\begin{remark}
Notice that having Galois group of order $12$ for a polynomial of degree $6$ is highly non-generic, as the Galois group of a random degree $6$ polynomial will be $\mathcal S_6$, which has cardinality $720$. Moreover, the  polynomial $f-t$ has a small Galois group even among all polynomials which are compositions of degree $3$ and degree $2$ polynomials: the generic condition is in fact to have a Galois group of size $72$. To see this, it is enough to notice that in the proof of Theorem \ref{thm:degree_6_pred} when we add all the roots of $H_1$, then $H_2$ splits because of the particular form of $f$.
More in general, notice that the index of $G_f$ in $\cS_{\deg(f)}$ gives essentially the measure of the advantage of using a polynomial with Galois group $G_f$ instead of one given by \cite[Proposition 3.4]{tamo2014family}.
\end{remark}

\section{Examples}\label{sec:examples}

In this section we see how the estimates given in our results agree with the actual 
number of totally split places. The computations were performed in  SAGE \cite{sagemath} and the code is available upon request.

Tables \ref{table_of_x3_smt_2} and \ref{table_of_x3_smt_5} show the agreement of the estimate in Theorem \ref{thm:degree_3_pred} with the actual number of totally split places. This agreement is asymptotic in $q$, but one can see some small perturbations due to the fact that $g_f=1$ for the polynomial we looked at. 
The polynomial lead to LRC optimal codes with locality $2$.

\begin{table}[H]
\caption{}
\begin{subtable}{5 cm}
\begin{tabular}{c|c|c}
	$q=2^{3n}$ & $\#T^1_{\text{split}}$ & lower bound \\
	\hline               
      8&     1&    0\\                  
     64&    10&    8\\                  
    512&    85&   77\\                  
   4096&   682&  661\\                  
  32768&  5461& 5401\\                  
\hline
\end{tabular}
\caption{Let $a\in \vF_8$ be a zero of $X^3+X+1$. The table compares $\#T^1_{\text{split}}(X(X+1)(X+a))$ with the lower bound given by Theorem \ref{thm:degree_3_pred} over some base fields}\label{table_of_x3_smt_2}
\end{subtable}
\quad \quad \quad \quad \quad\quad\quad
\begin{subtable}{5 cm}
\begin{tabular}{c|c|c}
	$q=5^n$ & $\#T^1_{\text{split}}$ & lower bound \\
	\hline               
     5&     1&     0\\
    25&     3&     2\\
   125&    21&    17\\
   625&   103&    95 \\
  3125&   521&   502 \\
 15625&  2603&  2562\\
 78125& 13021& 12927\\                                   

   \hline
\end{tabular}
\caption{Comparing $\#T^1_{\text{split}}(X(X+1)(X+3))$ with the lower bound given by Theorem \ref{thm:degree_3_pred} over many base fields}\label{table_of_x3_smt_5}
\end{subtable}
\end{table}

Table \ref{table_of_x4x2} shows almost a perfect agreement, which comes from the genus zero condition, which makes Chebotarev's error term a $O(1)$. The polynomial lead to LRC optimal codes with locality $3$.

In practice, for a random polynomial having a totally split place, the estimate  $q/(r+1)!$ for the number of totally split places seem to hold most of the time. 
Table \ref{table_of_x5_smt} shows the behaviour of a totally split polynomial of degree $5$. The polynomial lead to LRC optimal codes with locality $4$.

\begin{table}[H]
\caption{}
\begin{subtable}{5 cm}
\begin{tabular}{c|c|c}
	$q$ & $\#T^1_{\text{split}}$ & lower bound \\
	\hline               
  125& 15& 14\\                  
  127& 15& 14\\                  
  131& 16& 15\\                  
  137& 16& 16\\                  
  139& 17& 16\\                   
  149& 18& 18\\                  
  151& 18& 17\\                  
   \hline
\end{tabular}
\caption{Comparing $\#T^1_{\text{split}}(X^4+7X^2)$ with the prediction given by Theorem \ref{thm:x4x2prediction} over some base fields}\label{table_of_x4x2}
\end{subtable}
\quad \quad \quad \quad \quad\quad\quad
\begin{subtable}{5 cm}
\begin{tabular}{c|c|c}
	$q$ & $\#T^1_{\text{split}}$ & prediction: $\left\lceil{\frac{q+1}{(r+1)!}}\right\rceil$ \\
	\hline
 1787& 17& 15\\
 1789& 11& 15\\
 1801& 17& 16\\
 1811& 15& 16\\
 1823& 17& 16\\
 1831& 21& 16\\
 1847& 17& 16\\
 1849& 23& 16\\
 1861& 11& 16\\
   \hline
\end{tabular}
\caption{Comparing $\#T^1_{\text{split}}(X(X-1)(X-2)(X-3)(X-4))$ with the lower bound given by Theorem \ref{thm:pessimistic} over some base fields.}\label{table_of_x5_smt}
\end{subtable}
\end{table}

Tables \ref{table_of_x6} and \ref{table_of_x6_non_prime} show the agreement between the lower bound of Theorem \ref{thm:degree_6_pred} and the actual number of totally split places of $f$.

\begin{table}[H]
\caption{}
\begin{subtable}{5 cm}
\begin{tabular}{c|c|c}
	$q$ & $\#T^1_{\text{split}}$ & lower bound \\
	\hline               
 241& 20& 16\\
 263& 22& 18\\
 313& 26& 22\\
 347& 29& 24\\
 349& 29& 25\\
 359& 30& 25\\
 397& 33& 28\\

   \hline
\end{tabular}
\caption{Comparing $\#T^1_{\text{split}}((X^3+7X)^2)$ with the lower bound given by Theorem \ref{thm:degree_6_pred} over some prime base fields}\label{table_of_x6}
\end{subtable}
\quad \quad \quad \quad \quad\quad\quad
\begin{subtable}{5 cm}
\begin{tabular}{c|c|c}
	$q$ & $\#T^1_{\text{split}}$ & lower bound \\
	\hline                                                 
343& 28& 24\\ 
2197& 182& 174\\
16807& 1400& 1378\\  
   \hline
\end{tabular}
\caption{Comparing $\#T^1_{\text{split}}((X^3+5X)^2)$ with the lower bound given by Theorem \ref{thm:degree_6_pred} over some large non prime base fields}\label{table_of_x6_non_prime}
\end{subtable}

\end{table}

\begin{remark}
The reader should notice how fixing the polynomials in Tables \ref{table_of_x4x2} \ref{table_of_x6} and \ref{table_of_x6_non_prime} leads to optimal codes of large length and dimension, independently of the arithmetic properties of the pair $(q,r)$. 
See for example the case of Table \ref{table_of_x4x2} for $q=151$. We want to build an LRC of locality $3$ over $\vF_{151}$. As already observed just before the statement of Theorem \ref{thm:x4x2prediction}, with the arithmetic constraint $q\equiv 3 \mod 4$, it is not possible to use the explicit constructions listed in \cite[Section III]{liu2018new}. So we are left with \cite[Proposition 3.4]{tamo2014family} with $r=3$. We compute $\lceil{151\choose 4}/151^3\rceil=7$ which only ensures a $(3,7)$-good polynomial, and therefore using Theorem \ref{thm:LRCconstruction} we get LRCs with parameters $(4\ell,3t,3)$, with $\ell\leq 7$ and $t\leq \ell$, so these LRCs over $\vF_{151}$ have length up to $28$ and dimension up to $21$.

By using our method we get a codes of much larger dimension and length. In this case, using the polynomial $X^4+7X^2$ and Theorem \ref{thm:x4x2prediction}, we get LRCs with parameters $(4\ell,3t,3)$, where now $\ell\in \{1,\dots, 17\}$ and $ t\leq \ell$, so we can get optimal LRCs of length up to $68$, and dimension up to $51$.

\end{remark}

\section{Further Research}

In this paper we constructed new LRCs by fitting the theory of good polynomials in a Galois theoretical context via Chebotarev Density Theorem. The theory we developed includes all the previous constructions and shows that producing new good polynomials can be reduced to a Galois theoretical problem over global function fields:
\begin{problem}
Find polynomials $f\in\vF_q[X]$ such that 
\begin{itemize}
\item the splitting field $M_f$ of $f-t$ over the rational function field $\vF_q(t)$ has full constant field $\vF_q$
\item the Galois group of $f-t$ over $\vF_q(t)$ is as small as possible when compared with the Galois group of polynomials of the same degree.
\end{itemize}
\end{problem}
The first point is needed to be able to find totally split places. The second point is needed to have large length and dimension of the LRC arising from the polynomial.
We solve some instances of this problem to show how effective and practical the method is, which in turn allows us to build new good polynomials over base fields where the known constructions do not work.
As an immediate  application of Theorem \ref{thm:LRCconstruction}, this allows to construct LRCs with large length and dimension, over a fixed base field.
Also, the method produces existential and constructive results  for any totally split polynomial of fixed degree over any large base field without arithmetic restrictions on the pair $(q,r)$, as Theorem \ref{thm:pessimistic} provides an effective version of \cite[Proposition 3.4]{tamo2014family}.

\section{Acknowledgements}

We would like to thank Federico Amadio Guidi and Violetta Weger for checking the preliminary version of this manuscript.
The author is thankful to the Swiss National Science Foundation grant number 171248 and 171249.

\bibliographystyle{plain}
\bibliography{biblio2}

\begin{thebibliography}{10}

\bibitem{ballico2016higher}
Edoardo Ballico and Chiara Marcolla.
\newblock Higher hamming weights for locally recoverable codes on algebraic
  curves.
\newblock {\em Finite Fields and Their Applications}, 40:61--72, 2016.

\bibitem{barg2015locally}
Alexander Barg, Itzhak Tamo, and Serge Vl{\u{a}}du{\c{t}}.
\newblock Locally recoverable codes on algebraic curves.
\newblock In {\em Information Theory (ISIT), 2015 IEEE International Symposium
  on}, pages 1252--1256. IEEE, 2015.

\bibitem{cadambe2013upper}
Viveck Cadambe and Arya Mazumdar.
\newblock An upper bound on the size of locally recoverable codes.
\newblock {\em arXiv preprint arXiv:1308.3200}, 2013.

\bibitem{ferraguti2018complete}
Andrea Ferraguti and Giacomo Micheli.
\newblock Complete classification of permutation rational functions of degree
  three over finite fields.
\newblock {\em arXiv preprint arXiv:1805.03097}, 2018.

\bibitem{gopalan2014explicit}
Parikshit Gopalan, Cheng Huang, Bob Jenkins, and Sergey Yekhanin.
\newblock Explicit maximally recoverable codes with locality.
\newblock {\em IEEE Trans. Information Theory}, 60(9):5245--5256, 2014.

\bibitem{kosters2017short}
Michiel Kosters.
\newblock A short proof of a {C}hebotarev density theorem for function fields.
\newblock {\em Mathematical Communications}, 22(2):227--233, 2017.

\bibitem{liu2018new}
Jian Liu, Sihem Mesnager, and Lusheng Chen.
\newblock New constructions of optimal locally recoverable codes via good
  polynomials.
\newblock {\em IEEE Transactions on Information Theory}, 64(2):889--899, 2018.

\bibitem{micheli2019selection}
Giacomo Micheli.
\newblock On the selection of polynomials for the {DLP} quasi-polynomial time
  algorithm for finite fields of small characteristic.
\newblock {\em SIAM Journal on Applied Algebra and Geometry}, 3(2):256--265,
  2019.

\bibitem{papailiopoulos2014locally}
Dimitris~S Papailiopoulos and Alexandros~G Dimakis.
\newblock Locally repairable codes.
\newblock {\em IEEE Transactions on Information Theory}, 60(10):5843--5855,
  2014.

\bibitem{papailiopoulos2013repair}
Dimitris~S Papailiopoulos, Alexandros~G Dimakis, and Viveck~R Cadambe.
\newblock Repair optimal erasure codes through {H}adamard designs.
\newblock {\em IEEE Transactions on Information Theory}, 59(5):3021--3037,
  2013.

\bibitem{prakash2012optimal}
N~Prakash, Govinda~M Kamath, V~Lalitha, and P~Vijay Kumar.
\newblock Optimal linear codes with a local-error-correction property.
\newblock In {\em Information Theory Proceedings (ISIT), 2012 IEEE
  International Symposium on}, pages 2776--2780. IEEE, 2012.

\bibitem{silberstein2013optimal}
Natalia Silberstein, Ankit~Singh Rawat, O~Ozan Koyluoglu, and Sriram
  Vishwanath.
\newblock Optimal locally repairable codes via rank-metric codes.
\newblock In {\em Information Theory Proceedings (ISIT), 2013 IEEE
  International Symposium on}, pages 1819--1823. IEEE, 2013.

\bibitem{stichtenoth2009algebraic}
Henning Stichtenoth.
\newblock {\em Algebraic function fields and codes}, volume 254.
\newblock Springer Science \& Business Media, 2009.

\bibitem{tamo2014bounds}
Itzhak Tamo and Alexander Barg.
\newblock Bounds on locally recoverable codes with multiple recovering sets.
\newblock In {\em Information Theory (ISIT), 2014 IEEE International Symposium
  on}, pages 691--695. IEEE, 2014.

\bibitem{tamo2014family}
Itzhak Tamo and Alexander Barg.
\newblock A family of optimal locally recoverable codes.
\newblock {\em IEEE Transactions on Information Theory}, 60(8):4661--4676,
  2014.

\bibitem{tamo2016optimal}
Itzhak Tamo, Dimitris~S Papailiopoulos, and Alexandros~G Dimakis.
\newblock Optimal locally repairable codes and connections to matroid theory.
\newblock {\em IEEE Transactions on Information Theory}, 62(12):6661--6671,
  2016.

\bibitem{sagemath}
{The Sage Developers}.
\newblock {\em {S}ageMath, the {S}age {M}athematics {S}oftware {S}ystem
  ({V}ersion 7.4)}, 2016.
\newblock {\tt http://www.sagemath.org}.

\bibitem{turnwald}
Gerhard Turnwald.
\newblock On {S}chur's conjecture.
\newblock {\em J. Austral. Math. Soc. Ser. A}, 58(3):312--357, 1995.

\end{thebibliography}
\end{document}